\documentclass[runningheads]{llncs}

\let\OrigSpnewtheorem\spnewtheorem

\title{Streaming algorithms for computing coresets and \texorpdfstring{$k$}{k}-median clustering in the Hamming space}
\titlerunning{Streaming Coresets and $k$-Median Clustering in Hamming Space}

\usepackage[all]{nowidow}
\usepackage{bbold}
\usepackage{lmodern}
\usepackage{xspace}
\usepackage[T1]{fontenc}
\usepackage{amsmath}
\usepackage{amssymb}

\usepackage{amsthm}

\spnewtheorem*{proof}{Proof}{\itshape}{\rmfamily}

\usepackage[colorlinks=true,linkcolor=blue,citecolor=blue]{hyperref}
\usepackage[nameinlink,capitalise]{cleveref}

\usepackage{nicefrac}
\usepackage{tikz}
\usetikzlibrary{arrows.meta,positioning,shapes.geometric}

\usepackage[boxruled,linesnumbered]{algorithm2e}
\usepackage[noadjust,sort]{cite}

\usepackage{thmtools}
\usepackage{thm-restate}
\usepackage[noend]{algpseudocode}
\usepackage{todonotes}

\newcommand{\absolute}[1]{\lvert#1\rvert}
\newcommand{\ceil}[1]{\lceil#1\rceil}
\newcommand{\Absolute}[1]{\left\lvert#1\right\rvert}
\newcommand{\hd}{\texttt{\textup{hd}}}

\DeclareMathOperator*{\argmin}{arg\,min}
\DeclareMathOperator*{\polylog}{polylog}
\DeclareMathOperator*{\poly}{poly}
\newcommand{\R}{\mathbb{R}}
\newcommand{\eps}{\varepsilon}
\newcommand{\jl}{\mathrm{JL}}
\newcommand{\weight}{w}
\newcommand{\X}{\mathcal{X}}
\newcommand{\OPT}{\mathtt{OPT}}

\newcommand{\U}{U}
\newcommand{\Oh}{\mathcal{O}}
\newcommand{\eucl}{\mathbb{d}_2}
\newcommand{\euclsq}{\eucl^2}
\newcommand{\tOh}{\tilde{\mathcal{O}}}

\OrigSpnewtheorem{fact}[theorem]{Fact}{\bfseries}{\itshape}
\crefname{fact}{Fact}{Facts}
\Crefname{fact}{Fact}{Facts}
\crefname{algocf}{Algorithm}{Algorithms}
\Crefname{algocf}{Algorithm}{Algorithms}

\author{Taha El Ghazi\inst{1} \and
Jonas Ellert\inst{2} \and
Chien-Chung Huang\inst{1} \and
Tatiana Starikovskaya\inst{1}}

\authorrunning{T. El Ghazi et al.}

\institute{CNRS \& DIENS, École normale supérieure de Paris, PSL Research University, Paris, France\\
\email{\{taha.elghazi,chien-chung.huang,tatiana.starikovskaya\}@ens.fr}\\
Department of Mathematics and Computer Science (IMADA), University of Southern Denmark, Odense, Denmark\\
\email{ellert@imada.sdu.dk}}

\begin{document}
	
\maketitle
\begin{abstract}
Clustering is one of the most fundamental tools in data analysis, allowing large datasets to be summarized by a small number of representative points. Given a metric space $(\X, \mathbb{d})$ and a set $S$ of $n$ points in this space, the continuous $k$-median clustering problem asks to find a set $C$ of $k$ points that minimizes the objective function $\sum_{s\in S} \mathbb{d}(s,C)$. When $\X = \Sigma^\ell$ is the set of strings of length $\ell$ and $\mathbb{d}$ is the Hamming distance, the continuous $k$-median clustering problem is known to be W[1]-hard when parameterized by~$k$. In this work, we present the first $(1+\eps)$-approximation algorithm for this problem with FPT runtime $2^{\poly(\eps^{-1},k)} \cdot n\ell \polylog n$. An additional feature of the algorithm is that it can be implemented in streaming, requiring only $\tOh_\eps(\ell k + k^2)$ space. As an auxiliary tool of independent interest, we show the first streaming algorithm for computing an $\eps$-coreset for continuous $k$-median clustering under the Hamming distance.

\keywords{Streaming Algorithms \and Coresets \and $k$-Median Clustering \and Hamming Distance \and Parameterized Approximation Algorithms}
\end{abstract}

\section{Introduction}
Clustering is one of the most fundamental tools in data analysis, allowing large datasets to be summarized by a small number of representative points.
As an optimization problem, it has been extensively studied in Euclidean and general metric spaces from the viewpoints of approximation algorithms, parameterized complexity, and streaming, where the points arrive one-by-one and one must account for all the space used. Given a metric space $(\X, \mathbb{d})$ and a set $X$ of $n$ points in this space, the continuous  $(k,z)$-clustering problem asks to find a set $C$ of $k$ points (``centers'') that minimizes the objective function $\sum_{x\in X} (\mathbb{d}(x,C))^z$. When $z=1$, the $(k,z)$-clustering problem is called the continuous \emph{$k$-median clustering} problem, and when $z=2$, it is called the continuous \emph{$k$-means clustering} problem. Below, we omit ``continuous'' for brevity.

For $L_p$ and general metric spaces, the complexity of this problem is rather well understood (see, e.g.,~\cite{DBLP:journals/siamcomp/Chen09,10.1145/3717823.3718299} and references therein). 
In contrast, for classical string metrics such as Hamming, edit, and Ulam distances, the algorithmic landscape remains far less developed. String distances, and in particular the edit distance, are often used as similarity measures in bioinformatics, where clustering biological sequences is an increasingly important task. For instance, clustering is employed to reduce redundancy, bin sequencing reads, and group protein sequences into homologous families (see, e.g.,~\cite{10.1093/bib/bbs035,10.1145/3166072.3166076,10.1093/bib/bby090}).

Recall that the edit distance between two strings is the smallest number of character insertions, deletions, and replacements that are required to transform one string into the other. For $k = 1$, clustering under edit distance is also known as the median string problem~\cite{KOHONEN1985309}, and even this special variant of the problem is known to be NP-complete~\cite{DELAHIGUERA200039}, and remains so even on finite alphabets~\cite{10.1007/3-540-44888-8_23}. This hardness result does not exclude efficient approximation algorithms, but, to the best of our knowledge, none are known. 

This motivated Chakraborty, Das, and Krauthgamer~\cite{doi:10.1137/1.9781611976465.48} to study the Ulam distance, a special variant of the edit distance. 
The Ulam distance is defined on permutations over $\{1, \dots, \ell\}$ as the smallest number of character-move operations required to transform one string into another. 
While solving the \mbox{$1$-median} clustering problem under the Ulam distance exactly is NP-hard \cite{DBLP:conf/esa/FischerGH025}, 
\cite{doi:10.1137/1.9781611976465.48} showed a $(2-\delta)$-approximation algorithm for $1$-median clustering under the Ulam distance, for small constant $\delta$. Later, Chakraborty, Das and Krauthgamer \cite{DBLP:conf/innovations/Chakraborty0K23} gave a randomized $1.9999995$-approximation algorithm in streaming for the $k$-median clustering problem, using $k^2 \ell \polylog(n\ell)$ space and $(k \log (n\ell))^{O(k)} \ell^3$ update time (which is not FPT in $k$). Jaiswal, Kumar, and Yadav~\cite{DBLP:journals/corr/abs-2502-07653} then gave the first \emph{offline} $(2-\delta)$-approximation algorithm for the $k$-median clustering problem under the Ulam distance, in FPT time $\Oh((2k)^k n \ell)$.

Another common stepping stone for the edit distance is the Hamming distance (see, e.g.,~\cite{DBLP:conf/icalp/Bhattacharya023} who reduced the problem of streaming pattern matching under the edit distance to that under Hamming distance, or \cite{9317938}, where a framework for pattern matching under Hamming distance was extended to that under the edit distance). The Hamming distance between two equal-length strings is defined as the number of mismatches between them. The $k$-median problem under Hamming distance is already computationally hard. Under the Unique Games Conjecture, it is NP-hard to approximate within a factor better than $1.21$~\cite{cohenaddad:hal-02360762}. Moreover, 
W[1]-hardness with respect to $k$~\cite{DBLP:conf/fsttcs/FominGS19} suggests that an exact FPT algorithm is unlikely. 
In the offline setting, the best known result for the Hamming distance is by Ostrovsky and Rabani~\cite{DBLP:conf/focs/OstrovskyR00}, who gave a $(1+\eps)$-approximation algorithm in $\ell^{k^2} n^{\Oh(\eps^{-1} k^2)}$ time, but no $(1+\eps)$-approximation algorithm with running time $f(k,\eps)\cdot (n\ell)^{\Oh(1)}$ was previously known. Such approximation algorithms are referred to as \emph{efficient parameterized approximation schemes (EPAS)}~\cite{10353074,10.1093/comjnl/bxm048}. 

\paragraph*{Related Work.} In the related (yet quite different in nature) consensus string problem, the goal is to find a string $A$ that minimizes the objective function $\max_{B\in X} \mathbb{d}(B,A)$. This NP-hard problem has been studied for both the Hamming and the edit distance, and in particular in the parameterized setting. We refer the reader to a survey~\cite{DBLP:journals/eatcs/BulteauHKN14}, see also~\cite{bulteau_et_al:LIPIcs.MFCS.2018.1,GASIENIEC2004289,10.1007/978-3-319-07566-2_1,DBLP:journals/algorithmica/AmirPR16,10.1007/978-3-642-13509-5_28,10.1007/978-3-032-05228-5_12,AMIR2013371,10.1007/978-3-642-03784-9_23}. However, the techniques developed for this problem do not seem to extend to $(k,z)$-clustering directly.

\subsection{Results \& Overview}
In this work, we present the first EPAS for the $k$-median clustering problem under the Hamming distance. Using standard alphabet reduction techniques, one can reduce the problem over an arbitrary alphabet $\Sigma$ to the binary alphabet $\{0, 1\}$. For this alphabet, the squared $L_2$-distance coincides with the Hamming distance, i.e., $k$-means clustering after the reduction coincides with $k$-median clustering before the reduction. Hence, one can construct an $\eps$-coreset of size independent of $n$ by applying the recent offline algorithm of Huang and Vishnoi~\cite{HuangV20}. Intuitively, an $\eps$-coreset is a \emph{weighted} set of points that can be used to find a $(1+\eps)$-approximation of the minimum of the objective function (see \cref{def:coreset} for a precise definition). Brute-forcing over all possible partitions of the coreset into $k$ subsets, we obtain an EPAS. 

However, modern sequencing platforms generate biological data at rates that often exceed practical storage capacity, making streaming algorithms a natural fit. Motivated by this, we develop an EPAS that works in the streaming model. Recall that in this model, we assume that the strings arrive as a stream, one string at a time, and we account for all the space used, including the space needed to store any information about the input. Upon the arrival of each new string, the algorithm updates its state and outputs a $(1+\eps)$-approximation solution for $k$-median clustering for the current set of strings. This requirement of producing an output after every update is natural in bioinformatics applications where the databases are constantly updated, but stronger than the standard formulation of streaming clustering, where the solution is only required after processing the entire input (see, e.g.,~\cite{10.1007/s10994-024-06540-z}). Finally, we require each reported solution to be correct with high probability, ensuring an overall constant success probability.

The framework of our streaming algorithm follows that of the offline one: for each new point, we construct a small coreset, and then solve the problem on the coreset by brute-force. We note that this framework is by now standard and has in particular been used for streaming clustering algorithms before, see, e.g.,~\cite{10.1007/s10994-024-06540-z}. Our contribution lies in the implementation details specific to the Hamming distance. First, similarly to the offline EPAS above, we apply standard alphabet and then dimension reductions. 
On the reduced strings, we construct the \mbox{$\eps$-coreset} by applying a streaming algorithm for $k$-means in the $L_2$-space. (The current best algorithm with fully analyzed update time is that of Braverman et al.~\cite{braverman_et_al:LIPIcs.APPROX-RANDOM.2019.62}.\footnote{The more recent algorithm of Cohen-Addad et al.~\cite{10353181} is more space-efficient, but they do not analyze the update time.}) This already results in a relatively small coreset, but its size linearly depends on $\log n$, which intuitively is a critical obstacle towards an EPAS. We show that the algorithm of Huang and Vishnoi~\cite{HuangV20} which was designed to work on \emph{unweighted} sets can be modified to work on weighted sets as well, and apply it to the small-size coreset constructed by the streaming algorithm to further reduce its size without increasing the time complexity. 
For a stream of $n$ strings over an integer alphabet ${\Sigma = \{1, \dots n^{\Oh(1)}\}}$ (assumed throughout the paper), we obtain:

\begin{restatable}[Simplified]{theorem}{coreset}\label{cor:small-streaming-coreset}
For continuous $k$-median clustering in $(\Sigma^\ell, \hd)$, there is a streaming algorithm that maintains an $\eps$-coreset  on a stream of $n$ points while requiring $\tOh_\eps\!\left(\ell k + k^2 \log^2 \frac{1}{\delta}\right)$ space and $\tOh_\eps\!\left(\ell + k^2 \log^2 \frac1\delta\right)$ update time.\footnote{Hereafter, $\tOh_\eps$ hides $\eps^{-\Oh(1)} \polylog n$ factors.} The coreset is a subset of the input points of size $\Oh\!\left(\eps^{-6}k \log k \log \frac{k}{\eps \delta} \right)$ and the algorithm succeeds at any point of the stream with probability at least $1 - \delta$.
\end{restatable}

\def\coresetfullversion{
\begin{restatable*}[Full version]{theorem}{coreset}\label{cor:small-streaming-coreset}

Let $\ell,k,n \in \mathbb{Z}_+$ with $n \geq k$ and $\eps, \delta \in (n^{-\Oh(1)}, 1/2)$.
For continuous $k$-median clustering in $(\Sigma^\ell, \hd)$, there is a streaming algorithm that maintains an $\eps$-coreset on a stream of $n$ points while requiring $\tOh_\eps\!\left(\ell k + k^2 \log^2 \frac{1}{\delta}\right)$ space and $\tOh_\eps\!\left(\ell + k^2 \log^2 \frac1\delta\right)$ update time.\footnote{Recall that $\tOh_\eps$ hides $\eps^{-\Oh(1)} \polylog n$ factors.} The coreset is a subset of the input points of size $\Oh\!\left(\eps^{-2}k \log k \cdot \min(\log n, \eps^{-4}\log \frac{k}{\eps\delta}) \right)\!$, and the algorithm succeeds at every point of the stream with probability at least $1 - \delta$.
\end{restatable*}
}

Since an $\eps$-coreset approximates the objective function on the original ienput, one can apply any clustering algorithm to the coreset to obtain a $(1+\eps)$-approximation solution for the $k$-median problem on the full instance. 
However, to obtain an EPAS by brute-force, the coreset size must be independent of $n$. 
In this regime, the construction from \cref{cor:small-streaming-coreset} fails only with constant probability~$\delta$, and consequently, the resulting clustering is correct only with constant probability. Unlike standard randomized approximation schemes, this failure probability cannot be reduced by performing independent repetitions and selecting the best solution: 
the output is a set of centers rather than a single numerical estimate, and there is no efficient procedure to certify whether a given set of centers achieves a $(1+\eps)$-approximation without essentially solving the problem itself. 

To overcome this obstacle, we use a \emph{two-level coreset amplification} technique. 
Since the coresets in \cref{cor:small-streaming-coreset} are subsets of the input (a crucial feature), we can verify our solution as follows: We generate one ``large'' coreset of size $\poly(\eps^{-1}, k)\cdot \log n$, succeeding with high probability\footnote{Throughout the paper, with high probability, or w.h.p., means probability $1-n^{-c}$ for a constant $c \ge 1$ that can be chosen freely.}, and $\Theta(\log n)$ ``small'' coresets of size $\poly(\eps^{-1}, k)$, each succeeding with constant (non-zero) probability.
By a standard argument, at least one of the small coreset constructions succeeds with high probability. For each small coreset, we can construct a candidate set of centers in a brute-force manner efficiently, while the large coreset plays the role of ``certificate'', allowing us to identify a small coreset for which the construction succeeded.
This leads to the main result of the paper:
\def\kmedianHamfull{
\begin{restatable*}[Full version]{theorem}{kmedianHam}
\label{th:kmedianHam}
Let $\ell, k, n \in \mathbb{Z}_+$ with $n \geq k$, and let $\eps, \delta \in (n^{-\Oh(1)}, 1/2)$. There is a streaming algorithm that maintains a $(1+\eps)$-approximate solution for $k$-median clustering on a stream of $n$ points in $(\Sigma^\ell, \hd)$, while requiring $\tOh_\eps(\ell k + k^2)$ space and $\tOh_\eps(\ell \cdot 2^{\poly(\eps^{-1},k)})$ update time. It succeeds, w.h.p., at every point of the stream.
\end{restatable*}
}

\begin{restatable}[Simplified]{theorem}{kmedianHam}
\label{th:kmedianHam}
There is a streaming algorithm that maintains a $(1+\eps)$-approximate solution for $k$-median clustering on a stream of $n$ points in $(\Sigma^\ell, \hd)$, while requiring $\tOh_\eps(\ell k + k^2)$ space and $\tOh_\eps(\ell \cdot 2^{\poly(\eps^{-1},k)})$ update time. It succeeds, w.h.p., at every point of the stream.
\end{restatable}

As a remark, while in this work we focused on $k$-median clustering, that is, $(k,1)$-clustering, which arguably is the most natural formulation for the Hamming distance (compare, for example, with the consensus string problem), our framework can be applied to achieve a streaming EPAS for any constant $z$ by replacing the algorithmic tools we used with their counterparts for $(k,2z)$-clustering~\cite{braverman_et_al:LIPIcs.APPROX-RANDOM.2019.62,HuangV20}.

\section{Preliminaries}
\label{sec:prelim}
For any set $\mathcal X$, positive integers $i \leq d$, and $d$-dimensional point $x \in \mathcal X^d$, we write $x[i]$ to refer to the $i$-th coordinate of $x$.
For a real number $p \ge 1$, the $L_p$-distance between two points $x,y \in \R^d$ is defined by $\mathbb{d}_p(x,y) = (\sum_{i = 1}^{d} \absolute{x[i]-y[i]}^p )^{1/p}$. For $p = 2$, we further denote $\euclsq(x,y) = (\eucl(x,y))^2 = \sum_{i = 1}^{d} \absolute{x[i]-y[i]}^2$. The Hamming distance between two strings $x,y \in \Sigma^\ell$ is denoted as $\hd(x,y)$ and is defined as the number of positions $i \in \{1, \dots, \ell\}$ such that $x[i] \neq y[i]$.

\paragraph*{Weighted Sets.} Given a metric space $(\X, \mathbb{d})$, a \emph{weighted} set is a pair $(X,w)$ such that $X \subseteq \X$ and $w$ is a function $w:X\to (0, \infty)$. Furthermore, we define the \emph{total weight} of $X$ as $w(X) := \sum_{x \in X}w(x)$. If $w$ is identically equal to $1$, we say that the set is \emph{unweighted} and refer to it as $X$.
We say that $(X,w)$ is a subset of another weighted set $(X',w')$ if $X \subset X'$.

\paragraph*{\boldmath Clustering in $(\X,\mathbb{d})$.\unboldmath} Consider a metric space $(\X, \mathbb{d})$ and a weighted set $(X,w)$. For $Q \subseteq \X$ and $x\in X$, we denote $\mathbb{d}(x, Q) := \min_{q\in Q} \mathbb{d}(x,q)$, and $\mathbb{d}(X,Q) := \min_{x\in X}\mathbb{d}(x,Q)$. Furthermore, given a positive integer $z$, we define the \emph{$z$-clustering cost of $x$ under $Q$} to be $\nu_{z,w}(x, Q) := w(x) \cdot (\mathbb{d}(x, Q))^z$ and the \emph{$z$-clustering cost of $X$ under $Q$} to be $\nu_{z,w}(X, Q) := \sum_{x\in X} \nu_{z,w}(x, Q)$. If neither of the parameters $z$ or $w$ is ambiguous, we may refer to $\nu_{z,w}$ as simply $\nu$.

Given a positive integer $k$, the goal of the (continuous) $(k,z)$-clustering problem is to minimize the quantity $\nu_{z,w}(X,C)$, where $C$ can be any subset of $\X$ of size at most $k$. When $z=1$, the $(k,z)$-clustering problem is called the \emph{$k$-median clustering} problem, and when $z=2$, it is called the \emph{$k$-means clustering} problem.

\begin{definition}[Coreset for $(k,z)$-clustering]
	\label{def:coreset}
	For a metric space $(\X, \mathbb{d})$, let $(X,w)$ be a weighted set of $n$ points in $\X$. A weighted set $(S,u)$ is an \emph{$\eps$-coreset for $(k,z)$-clustering} if
	$\absolute{\nu_{z,u}(S,C)-\nu_{z,w}(X,C)} \le \eps \nu_{z,w}(X,C)$ 
    for all sets $C \subseteq \X$ satisfying $\absolute{C} \le k$.\footnote{Such coresets are sometimes called \emph{strong}.}	
\end{definition}

\begin{fact}[{\cite[Theorem 3 with $\rho = 2$ and $\delta = n^{-c}$ for constant $c$]{braverman_et_al:LIPIcs.APPROX-RANDOM.2019.62}}]\label{fact:means_clustering}
Let $d, k$ be positive integers, and let $\eps \in (0, 1)$.
For $k$-means clustering in $(\R^d, \eucl)$, there is a streaming algorithm that maintains an $\eps$-coreset of size $\Oh(\eps^{-2}k\log k \log n)$ on a stream of $n$ points while requiring $\Oh(\eps^{-2}dk\log k \log^2 n)$ space and update time. It succeeds, w.h.p., at every point of the stream.
\end{fact}

As mentioned in the introduction, this algorithm results in a coreset of size dependent on~$n$, which is an obstacle to achieving FPT time. To remedy this, we give the following offline coreset construction algorithm for weighted sets, obtained by modifying the construction for unweighted inputs given in~\cite[Theorem~5.1]{HuangV20}.

\begin{restatable}{theorem}{smallcoresetoffline}
	\label{thm:small-coreset-offline}
	Let $d, k$ be positive integers, and let $\eps, \delta \in (0, 1/2)$.
    There exists a randomized
	algorithm that, for a given weighted set $(X,w)$ of $n > k$ points in $(\R^d,\eucl)$ with integer weights, with probability at least $1-\delta$, constructs an $\eps$-coreset for $k$-means clustering of size
	$\Oh\!\left(\eps^{-6}k \log k \log \frac{k}{\eps \delta} \right)$ 
	and runs in time and space
	$\Oh\!\left(ndk + nd \log\frac{n}{\delta} + k^2 \log^2\frac{1}{\delta} \log^2 n + n \log W \max(k, \log n)\right)\!$,
	where $W$ is the sum of the weights of points in $X$.
	Furthermore, the coreset is a subset of $X$.\footnote{Despite this additional property, the coreset is for \emph{continuous} $(k,z)$-clustering.} The algorithm works under the standard assumption that a lower bound on the $\eucl$ distance between two points in $X$ is known.
\end{restatable} 

\paragraph*{Model of Computation.} We consider the streaming model of computation, in which the input points arrive one by one. The stream consists of $n$ points from $\Sigma^\ell$ with $\Sigma = \{1, \dots, \absolute{\Sigma}\}$ and $\absolute{\Sigma} = n^{\Oh(1)}$, which are to be processed by a word RAM using words of size $\Theta(\log (n\ell))$. 
The space complexity measures the total amount of memory used, including the space required to store (parts of) the input, and is expressed in terms of the number of memory words. 
As is common in the literature on streaming algorithms, we do not account for the preprocessing time (before the first point arrives).

\section{Streaming Coreset Computation via an Embedding}
\label{sec:coreset}
In this section, we consider the metric space $(\Sigma^\ell, \hd)$.
We denote by $x_1, \dots, x_n$ the elements of the input set $X$ that arrives as a stream. 
We describe a streaming algorithm that computes an $\eps$-coreset for $k$-median clustering for $X$, along with some auxiliary information that will be useful in the next section.

The algorithm consists of four main steps.
In the first two steps, it reduces the problem to low-dimensional Euclidean space. Using Karloff’s embedding~\cite{KARLOFF}, 
it maps strings over~$\Sigma$ to binary vectors, with the guarantee that the (scaled) squared 
$\eucl$ distance between embedded points provides a $(1+\eps)$-approximation 
of their Hamming distance (Step~1). For our polynomial-size alphabet, however, this embedding increases the dimension 
by a factor of $\Oh(\log^2 n)$, which is prohibitive in the streaming setting because
the space complexity depends linearly on the dimension. To mitigate this, the algorithm applies
a terminal embedding~\cite[Theorem~1.11]{10.1145/3313276.3316307} that extends the Johnson--Lindenstrauss transform. This reduces the dimension 
while preserving all distances to candidate centers up to an additional multiplicative 
$(1+\eps)$ distortion (Step~2). The composition of these embeddings therefore yields a $(1+\Oh(\eps))$-distortion mapping from Hamming space to low-dimensional $L_2$ space. In Step~3, the algorithm applies a 
streaming coreset construction for $k$-means clustering in Euclidean space 
(see, e.g., \cref{fact:means_clustering}). By the distortion guarantee, the resulting 
coreset is also a $(1+\Oh(\eps))$-coreset for $k$-median on the original set of strings. Finally, in Step 4, if necessary, we shrink the size of the coreset (increasing the error probability) using \cref{thm:small-coreset-offline}. Below, we provide the full details.

\paragraph*{Tools for Step 1: Reduction to Binary Alphabet.}
We first build an embedding from $\Sigma$ to binary alphabets, using techniques introduced by Karloff~\cite{KARLOFF}.\footnote{Note that a simpler unary reduction gives an approximation factor of $2$, which is too high for our purposes. The Kopelowitz--Porat embedding~\cite{kopelowitz_et_al:OASIcs.SOSA.2018.10} gives an approximation factor of $1+\eps$ and improves the dimensional blow-up of Karloff's embedding, but there is no direct way to adapt their techniques to the streaming setting.}
Below, we show how Karloff's algorithm can be implemented in small space.

\begin{definition}[Balanced matrices, {\cite{KARLOFF}}]
\label{def:balanced_matrix}
Let $\xi \in (0, 1]$, and let $r$ and $M$ be positive integers. An $r \times M$ binary matrix $A = (a_{ij})$ is $\xi$-balanced if, for all $j, j' \in \{1, \dots, M\}$, the number of $i \in \{1, \dots, r\}$ such that $a_{ij} \neq a_{ij'}$ is in $[(1-\xi)r/2, (1+\xi)r/2]$.
\end{definition}

\begin{fact}[{\cite[Theorems 6 and 8]{KARLOFF}}]
\label{fact:matrix_constr}
Consider an integer $M$ such that $\absolute{\Sigma} \le M \equiv 0 \pmod 3$ and let $\xi \in [1/M, 1)$ be a constant. There is an $\Oh(\xi^{-2} \cdot M \cdot \log^3 M)$-time algorithm that constructs a matrix $B$ of size $\Oh(\xi^{-2} \cdot \log^2(M+1)) \times \Oh(\log (M+1))$ and a matrix $L$ of size $\Oh(\log (M+1)) \times M$ with the following properties. The $j$-th column of $L$ is the binary representation of $j$ with the least significant bit at the bottom. The matrix $A = B \cdot L \pmod 2$ is $\xi$-balanced.
\end{fact}

\begin{fact}[{cf.~\cite[Theorem 1]{KARLOFF}}]
\label{fact:embedding}
Let $ \eps \in (0, 1/2)$ be a constant. Consider an $\eps/3$-balanced binary matrix $A = (a_{ij})$ of size $r \times M$, where $\absolute{\Sigma} \le M = 0 \pmod 3$. For a string $x \in \Sigma^\ell$, define $\kappa(x) := y \in \{0,1\}^{r \cdot \ell}$, where, for $i \in \{1, \dots, r\}$ and $j \in \{1, \dots, \ell\}$, we have $y[i\ell - \ell + j] = a_{ix[j]}$.
For strings $x,y \in \Sigma^\ell$ we then have:

\begin{equation}
\label{eq:alphabet_reduction}
\hd(x,y) \le \frac {\hd(\kappa(x),\kappa(y))}{\frac r 2 - \frac{\eps r} 6} \le (1+\eps) \cdot \hd(x,y)
\end{equation}
\end{fact}
\begin{restatable}{corollary}{streaming}
\label{cor:streaming}
Given a string $x \in \Sigma^\ell$, consider an integer $M$ such that $\absolute{\Sigma} \le M = 0 \pmod 3$ and let $ \eps \in (1/M, 1/2)$. There is a streaming algorithm that uses $\Oh(\eps^{-2}\log^3 M)$ space and computes in $\Oh(\log M)$ time per character an embedding $\kappa(x): \Sigma^\ell \rightarrow \{0,1\}^{r \ell }$ which satisfies both $r = \Oh(\eps^{-2} \log^2 (M+1))$ and \cref{eq:alphabet_reduction} for all strings $x,y \in \Sigma^\ell$.
\end{restatable}
\begin{proof}
We apply \cref{fact:matrix_constr} to construct the matrices $B,L$ defining the embedding $\kappa$ of \cref{fact:embedding} in $\Oh(\eps^{-2} M \log^3 M)$ time. Note that it is enough to store $B$ in $\Oh(\eps^{-2}\log^3 M)$ space as $L$ can be computed on the fly. An entry $a_{ij}$ of the matrix $A = B \cdot L \pmod 2$ can then be computed in $\Oh(\log M)$ time.
\end{proof}

\paragraph*{Tools for Step 2: Dimension Reduction.}
\begin{fact}[{Johnson--Lindenstrauss transform~\cite[Theorem 1.1]{ACHLIOPTAS03}}]
\label{fact:JL}
Let $X$ be a set of $n$ points in~$\R^m$. Given $\eps, \beta > 0$, let
\[k_0 = \frac{4+2\beta}{\eps^2-\eps^3/3} \log n\]
For integer $k \ge k_0$, let $R$ be an $m \times k$ random matrix with $R(i,j) = r_{i,j}$, where $\{r_{i,j}\}$ are independent random variables chosen from the following distribution:
\[r_{i,j} = 
\begin{cases}
+1 \text{ with probability } 1/2,\\
-1 \text{ with probability } 1/2.
\end{cases}
\]  
For $x \in X$, define $\jl(x) = \frac{1}{\sqrt{k}} x^T R$. With probability at least $1-n^{-\beta}$, for all $x,y \in X$, 
\begin{equation}
\label{eq:jl}
(1-\eps) \cdot \eucl(x,y) \le \eucl(\jl(x),\jl(y)) \le (1+\eps) \cdot \eucl(x,y)
\end{equation}
\end{fact}

\begin{definition}[Terminal embedding {\cite{10.1145/3188745.3188828}}]
Suppose that we are given a set of points $X \subset \R^m$ (which we call \emph{terminals}), where $m$ is a positive integer. We say that a map $f: \R^m \rightarrow \R^{m'}$ is a \emph{terminal dimension reduction} with distortion $D$ if for every terminal $x \in X$ and point $y \in \R^m$ ($y$ may be a terminal), we have $\eucl(x,y) \le \eucl(f(x),f(y)) \le D \cdot \eucl(x,y)$.
\end{definition}

\begin{fact}[{\cite[Theorem 1.11]{10.1145/3313276.3316307}}]
\label{fact:terminal}
Let $X$ be a set of $n$ points in $\R^m$ and $ \eps \in (0, 1/2)$ a parameter. Assume that a function $\jl : \R^m \rightarrow \R^{m'-1}$, for some integer $m'$, satisfies \cref{eq:jl}. In this case, there exists a terminal dimension reduction $\gamma : \R^m \rightarrow \R^{m'}$ with distortion $1+\eps$, where $m' = \Oh(\frac{\log n}{\eps^2})$. For $x \in X$, $\gamma(x) = (\jl(x), 0)$ (i.e. $\gamma(x)$ equals $\jl(x)$ on the first $m'-1$ coordinates and $0$ on the $m'$-th coordinate).
\end{fact}

\paragraph*{Combining the Pieces: Coreset for \boldmath$k$\unboldmath-Median Clustering.}
Recall that we receive the elements $x_1, \ldots, x_n$ of $X$ as a stream. 

For each $x_i$, we perform the following steps (corresponding to Steps 1–4 of the high level description):
\smallskip
\begin{enumerate}
\item With \cref{cor:streaming}, we obtain $x_i' = \kappa(x_i) \in \{0,1\}^{r \ell}$, where $r = \Oh(\eps^{-2} \log^2 M)$. Here,~$M$ is a parameter that must satisfy $\absolute{\Sigma} \le M = 0 \pmod 3$ and $\eps > 1/M$.
Note that for $\eps \in (n^{-\Oh(1)}, 1/2)$, i.e., if there exists a constant $c > 0$ such that $\eps > n^{-c}$, we can choose $M = 3 \cdot \max(\ceil{n^c}, \absolute{\Sigma})$, which implies $\log M = \Oh(\log n)$.
\item Immediately after computing $x_i'$, we obtain $x_i'' = \tau(x_i') := \gamma(x_i') / {\sqrt{\frac r 2 - \frac {\eps r} 6}} \in \R^{\Oh(\log (n) / \eps^2)}$, where the factor $1/{\sqrt{\frac r 2 - \frac {\eps r} 6}}$ is the square root of the scaling factor from \cref{eq:alphabet_reduction}, and $\gamma$ is the terminal embedding for the set $X' = \{x_1', x_2',\ldots, x_n'\}$ with parameter $\eps$ (\cref{fact:terminal}). For computing $\gamma$, we multiply $x_i'$ by the matrix from \cref{fact:JL} and append a zero.
\item Immediately after computing $x_i''$, we feed it into the algorithm of \cref{fact:means_clustering} with parameter~$\eps$, which outputs an $\eps$-coreset $(U, w_U)$ for $k$-means clustering.
\item If $\absolute{U} > \eps^{-6}k \log k \log \frac{k}{\eps\delta}$ for a parameter $\delta \in (0, \frac 1 2)$, then we use $(U, w_U)$ to compute a smaller $\eps$-coreset in the following way. First, note that the description of the algorithm of \cref{thm:small-coreset-offline} in~\cite{braverman_et_al:LIPIcs.APPROX-RANDOM.2019.62} implies that for all $u \in U$ its weight $w_U(u)$ is a ratio $a/b$, where $a$ and $b$ are integers of order $n^{\Oh(1)}$. This lets us apply \cref{lem:integer-reweighting} to scale the weights so that they become integers in $n^{\Oh(1)}$, apply \cref{thm:small-coreset-offline} to compute an $\eps/2$-coreset in the scaled space, and then scale the weights back. We denote by $(Q, w_Q)$ the resulting set. Otherwise, we simply let $(Q, w_Q) = (U, w_U)$. Note that all points in $U$ are integer vectors multiplied by a fixed real number, so we may assume that we know a lower bound on the smallest distance between two points in $U$.
\end{enumerate}
\smallskip
For brevity, let $\mu(x) = \tau(\kappa(x))$ for a point $x \in \Sigma^\ell$, and let $\mu(X) = \{\mu(x) : x \in X\}$. 
In principle, the set $(Q, w_Q)$ is the desired coreset, but it resides in the space $\R^{\Oh(\log (n) / \eps^2)}$ \emph{after} applying the reductions in Steps 1 and 2. Hence, we show that we can indeed maintain the relevant original points from $\Sigma^\ell$:

\begin{lemma}\label{lem:store_original_space_points}
    We can maintain a set $S \subseteq \{x_1, \dots, x_i\}$ such that $Q = \{\mu(s) : s \in S \}$ in $\Oh(\eps^{-2}k\ell \log k \log n)$ space (and without time overhead).
\end{lemma}

\begin{proof}
    It suffices to show that we can maintain the original space representation of the set $U$, as $Q$ is always a subset of $U$.
    Looking at the implementation of \cref{fact:means_clustering} in \cite{braverman_et_al:LIPIcs.APPROX-RANDOM.2019.62}, the algorithm for maintaining $U$ is as follows: a set $\tilde U$ of ``interesting'' points is maintained (the union of the sets $M_y$ in \cite[Algorithm 1]{braverman_et_al:LIPIcs.APPROX-RANDOM.2019.62}), and we always have $U \subseteq \tilde U$. Notably, each point from $X$ enters the set $\tilde U$ exactly once (immediately after its arrival), and may be deleted from it later. Hence, it suffices to memorize each point $x_i$ on arrival, and discard it as soon as $\mu(x_i)$ is deleted from $\tilde U$. Then, the overall space usage is $\Oh(\ell)$ times the maximum size of $\tilde U$. Finally, we have $\absolute{\tilde U} = \Oh(\eps^{-2}k \log k \log n)$ by \cite[Lemma 15 with $\delta = n^{-c}$ for constant $c$]{braverman_et_al:LIPIcs.APPROX-RANDOM.2019.62}.
    We omit the details of the simple bookkeeping needed to delete each point at the right time.
\end{proof}

The final coreset in the original space, reported after each arriving point, is $(S, w_S)$ with~$S$ as defined above, and $w_S(s) = w_Q(\mu(s))$. The algorithm outputs $Q$ alongside $(S, w_S)$.
Towards the proof of \cref{cor:small-streaming-coreset}, we show that $(S, w_S)$ is indeed a coreset. Intuitively, this is the case because $\mu(x) = \tau(\kappa(x))$, where $\kappa$ maps the Hamming space to $L_2$ space with a small distortion, and $\tau$ reduces the dimension of the $L_2$ space while preserving, up to a small distortion, distances between points in $\kappa(X)$ and the rest of the space. Thus, the preimage of the coreset for $\mu(X)$ should result in a coreset for~$X$. 

\begin{restatable}{lemma}{coresetanalysis}\label{lm:coreset_analysis}
The set $(S, w_S)$ is an $\Oh(\eps)$-coreset for $k$-median clustering for the set $\{x_1, \dots, x_i\}$ with probability at least~$1-2\delta$.
\end{restatable} 

We will also need the following auxiliary result: 

\begin{restatable}{lemma}{weightcoresets}
\label{clm:weight-of-coresets}
Let $(X,\weight)$ be a weighted subset of $\R^d$ and $\eps \in (0, 1/2)$. If $(U, \weight_U)$ is an $\eps$-coreset of $X$ for $k$-means clustering under the $\eucl$ distance and $U \subseteq X$, then $\weight_U(U) \leq 6\weight(X)$.
\end{restatable}
\begin{proof}
Let $\Delta$ be the diameter of $X$, i.e., the maximal distance between a pair of points in~$X$. Let $q \in \R^{d}$ such that $\eucl(X, \{q\}) = 2\Delta$. We have $\Delta^2 w(X) \leq \nu(X, \{q\}) \leq 3\Delta^2 \weight(X)$ and $\Delta^2 \weight_U(U) \leq \nu(U, \{q\}) \leq 3\Delta^2 \weight_U(U)$, where the lower bounds are due to the definition of $\nu$ for $z = 2$, and the upper bounds are due to the triangle inequality. Furthermore, since $(U,\weight_U)$ is an $\eps$-coreset for $k$-means clustering, we have 
\[\Delta^2 \weight_U(U) \leq \nu(U, \{q\}) \leq (1+\eps)\cdot  \nu(X, \{q\})\leq 3\cdot (1+\eps) \cdot \Delta^2 \weight(X) < 6 \cdot \Delta^2 \weight(X)\]
(by considering a set $C = \{q\}$ in \cref{def:coreset}). As a result, $\weight_U(U) \leq 6\weight(X)$.
\end{proof}

\coresetfullversion
\begin{proof}
If the algorithm does not fail, then the obtained set is an $\Oh(\eps)$-coreset by \cref{lm:coreset_analysis}. Running the algorithm with parameter $\eps / c$ for a large enough constant $c$ instead of $\eps$ does not affect the claimed complexities and results in an $\eps$-coreset. By \cref{lm:coreset_analysis}, the failure probability after each point is at most $2\delta$. By running the algorithm with $\delta/2$ instead of $\delta$, we again do not affect the complexity and succeed with probability at least $1 - \delta$.

For the space complexity of the algorithm, we consider each component separately:
\smallskip
\begin{itemize}
    \item The matrix for Karloff's embedding can be stored in $\Oh(\eps^{-2} \log n)$ space (\cref{cor:streaming}) and the matrix for the Johnson--Lindenstrauss transform in $\Oh(\eps^{-2} \ell \log^3 n)$ space (\cref{fact:JL}).
    \item Computing a point $x_i''$ requires $\Oh(\eps^{-2} \log n)$ space (note that $x_i'$ does not have to be stored explicitly and can be generated on the fly in $\Oh(\log n)$ time per character, see \cref{cor:streaming}).
    \item The streaming coreset construction algorithm \cref{fact:means_clustering} uses $\Oh(\eps^{-3} k \log k \log^3 n)$ space.
    \item We need $\Oh(\eps^{-2} k\ell \log k \log n)$ space to maintain the interesting points with \cref{lem:store_original_space_points}.
    \item The offline coreset construction (\cref{thm:small-coreset-offline}) applied on $(U, \weight_U)$ uses space
    \[
\begin{aligned}
\Oh(&\eps^{-4} k^2 \log k \log^2 n + \eps^{-4} k \log k \log^2 n \log \frac{n}{\delta}  + k^2 \log^2 \frac{1}{\delta} \log^2 (\eps^{-2} k \log k \log n)  \\
&+ \eps^{-2} k \log k \log^2 n \max(k, \log n)) = \Oh\!\left(\eps^{-4} k^2 \log^2 \!\left(\frac{1}{\delta}\right)\log^4 n\right),
\end{aligned}
\]
    using $w(U) = \Oh(w(X)) = \Oh(n)$ (\cref{clm:weight-of-coresets}) and $k \le n$.
\end{itemize}
\smallskip
Hence, the algorithm uses $\Oh(\eps^{-2} \ell (k \log k + \log n) \log^2 n + \eps^{-4} k^2 \log^2 (\frac{1}{\delta}) \log^4 n)$ space. 
Finally, we analyze the update time of the algorithm:
\begin{itemize}
    \item Computing a point $x_i''$ requires $\Oh(\eps^{-2} \ell \log^4 n)$ time (recall that $x_i'$ is generated on the fly in $\Oh(\log n)$ time per character, see \cref{cor:streaming}); 
    \item The streaming coreset construction (\cref{fact:means_clustering}) uses $\Oh(\eps^{-4} k \log k \log^4 n)$ update time;
    \item The offline coreset construction (\cref{thm:small-coreset-offline}) applied on $(U, \weight_U)$ runs in time $\Oh\!\left(\eps^{-4} k^2 \log^2 \!\left(\frac{1}{\delta}\right)\log^4 n\right)$.
\end{itemize}
\smallskip
Hence, the total update time is $\Oh((\eps^{-2} \ell + \eps^{-4} k^2 \log^2 (\frac{1}{\delta})) \cdot \log^4 n)$, concluding the proof.\footnote{We apply the terminal embedding to remove the dependence on the original dimension in the time complexity of the coreset constructions. One could use a more efficient (and more complicated) dimension reduction (see~\cite{10.1145/3313276.3316350} and references therein), but this would only affect logarithmic factors in the runtime, which we do not attempt to optimize.}
\end{proof}

\section{\texorpdfstring{A Streaming FPT Algorithm for $k$-Median Clustering}{A Streaming FPT Algorithm for k-Median Clustering}}
We start with a $(1 + \Oh(\eps))$-approximation algorithm for $k$-median clustering, assuming that we are provided with an algorithm that constructs an $\eps$-coreset for $k$-median clustering. 

\label{sec:clustering}

\def\smalldelta{\nicefrac{1}{10}}
\def\bigdelta{n^{-c}}
\def\totaldelta{n^{-c+1}}
\begin{theorem}\label{claim:oneplusepskmeans}
Let $\ell,k$ be positive integers and $X$ be an unweighted set of $n$ points in $(\Sigma^\ell, \hd)$, where $\Sigma$ is an alphabet of size polynomial in $n$. Let $\eps \in (n^{-\Oh(1)}, 1/2)$. Suppose that for every $\delta \in (n^{-\Oh(1)},1/2)$, an $\eps$-coreset for $k$-median clustering can be maintained in streaming as follows. The coreset is of size at most 
$f(\delta) \geq k$, 
the update time per arriving point is at most $t(\delta)$, and the space does not exceed $s(\delta) \geq \ell \cdot f(\delta)$ at any time. After the arrival of any given point, the coreset is correct with probability at least $1-\delta$. 

There is a streaming algorithm for ${(1+\Oh(\eps))}$-approximate $k$-median clustering on~$X$, reporting a solution for the already seen points after each arriving point. It has update time $\Oh((t(\smalldelta) + 
\ell \cdot f(\smalldelta) \cdot k^{f(\smalldelta)} + k \ell \cdot f(\bigdelta)) \cdot \log n + t(\bigdelta))$, space complexity ${\Oh(s(\smalldelta) \cdot \log n + s(\bigdelta))}$, and all reported solutions are correct with high probability.
\end{theorem}
\begin{proof}
Fix an integer constant $c > 0$, with the goal of achieving success probability $1 - n^{-c + 2}$. We use the streaming algorithm to maintain one $\eps$-coreset $(U,\weight_U)$ with $\delta = \bigdelta$, and $m := c\lceil \log_{2} (n) \rceil$ different $\eps$-coresets $(U_1,\weight_{U_1}), \dots, (U_m,\weight_{U_m})$ with $\delta = \smalldelta$.
	After every arriving point, for every $i \in \{1, \dots , m\}$, we solve the $k$-median clustering problem on $U_i$ by brute-force: We enumerate all the partitions of $\U_i$ into $k$ disjoint subsets $\U_{i,1}, \dots, \U_{i,k}$. 
    These~$k$ subsets represent the clusters induced by $k$ implied centers $u_{i,1}, \dots ,u_{i,k}$. By definition, $u_{i,j} := \argmin_{u \in \Sigma^\ell} \sum_{ p\in \U_{i,j}} w_{U_i}(p) \cdot \hd(p, u)$, so we can compute the centers $u_{i,1}, \dots ,u_{i,k}$ as well as their clustering cost in time $\Oh(\ell \cdot f(\smalldelta))$ with a majority vote character by character. We then define $C_i := \{c_{i,1}, \dots, c_{i,k}\}$ to be the set of centers with the smallest clustering cost among all partitions of $U_i$. (We do not need to store all partitions of $U_i$ at the same time; we maintain $C_i$ directly while enumerating the partitions.)
Finally, we evaluate the clustering cost of $\U$ under each set of centers $C_1, \dots, C_m$, and we define the final set $C^\star$ of $k$ centers so that $\nu( \U, C^\star) = \min_{j \in [m]} \nu(\U, C_j)$.

\textsf{\bfseries Time.}
We maintain one coreset in $t(\bigdelta)$ update time, and $m=\Oh(\log n)$ coresets in $\Oh(t(\smalldelta)\log n)$ update time. For each $i$, due to $\absolute{U_i}\le f(\smalldelta)$, there are $\Oh(k^{f(\smalldelta)})$ possible partitions, and for each we compute the $k$ centers by coordinate-wise weighted majority vote in $\Oh(\ell\cdot f(\smalldelta))$ time, giving $\Oh(\ell\cdot f(\smalldelta)\cdot k^{f(\smalldelta)})$ time per $i$.
Finally, evaluating $\nu(U,C_i)$ for each candidate $C_i$ and the larger coreset of size $\absolute{U}\le f(\bigdelta)$ costs $\Oh(k\ell\cdot f(\bigdelta))$ time per $i$. Summing over the $\Oh(\log n)$ choices of $i$ results in the claimed update bound.

\textsf{\bfseries Space.}
Maintaining and storing the coresets takes $\Oh(s(\smalldelta)\log n + s(\bigdelta))$ space.
This dominates the space complexity, as it exceeds the $\Oh(k\ell \log n) \subseteq \Oh(s(\smalldelta) \cdot \log n)$ space needed to maintain the sets $C_i$.

\textsf{\bfseries Correctness.} Suppose that $U$ and at least one of $U_1, \dots, U_m$, say $U_j$, are correct coresets of~$X$. Furthermore, let $\OPT(U)$ be an optimal solution of $k$-median clustering with input $U$, and $\OPT(X)$ be an optimal solution of $k$-median clustering with input $X$. We have, by the definition of coresets, $\nu(U,C_{j}) \leq (1+\eps) \cdot \nu(X,C_{j})$ and $\nu(X,C_{j}) \leq \frac{1}{1-\eps} \cdot \nu(U_j,C_{j})$. Hence,
\begin{equation}\label{eq:star-1}
\nu(U,C_{j}) \leq (1+\eps) \cdot \nu(X,C_{j}) \leq \frac{1+\eps}{1-\eps} \cdot \nu(U_j,C_{j}).
\end{equation}
Furthermore, by optimality of $C_j$, we have 
\begin{equation}\label{eq:star-2}
\nu(U_j, C_j) \leq \nu(U_j, \OPT(U))
\end{equation} 
Again by the definition of coresets, we obtain 
\begin{equation}\label{eq:star-3}
\nu(U_j,\OPT(U)) \leq (1+\eps) \cdot \nu(X,\OPT(U)) \leq\frac{1+\eps}{1-\eps} \cdot \nu(U,\OPT(U)).
\end{equation}
Now recall that $\nu(\U, C^\star) = \min_{j' \in [m]} \nu(\U, C_{j'})$. We combine Equations~\ref{eq:star-1}, \ref{eq:star-2}, \ref{eq:star-3} to obtain 
\begin{equation}\label{star}
\nu(U,C^{\star}) \leq \nu(U,C_{j}) \leq  \left(\frac{1+\eps}{1-\eps}\right)^2 \cdot \nu(U,\OPT(U)). 
\end{equation}
We now apply the same reasoning for $X$: 
\begin{align*}
\nu(X,C^{\star}) & \leq \frac{1}{1-\eps} \cdot \nu(U,C^{\star})  &   & \text{(\cref{def:coreset})}\\
& \leq \frac{(1+\eps)^2}{(1-\eps)^3} \cdot \nu(U,\OPT(U))  &   & \text{(\cref{star})} \\
			& \leq \frac{(1+\eps)^2}{(1-\eps)^3} \cdot \nu(U,\OPT(X))  &   & \text{(optimality of $\OPT(U)$)}\\
			& \leq \frac{(1+\eps)^3}{(1-\eps)^3} \cdot \nu(X,\OPT(X))&   & \text{(\cref{def:coreset})}\\
			& \leq (1+\Oh(\eps)) \cdot  \nu(X,\texttt{OPT}(X)) &   & \text{for $\eps \leq 1/2$}
\end{align*}

Finally, observe that the probability of all $m$ small coresets failing is at most $(\smalldelta)^{m} \leq n^{-c}$, and the probability of the large coreset failing is at most $n^{-c}$. Hence, by the union bound, there is at least one point after which the construction fails with probability at most $n \cdot (n^{-c} + n^{-c}) \leq 2n^{-c + 1} \leq n^{-c + 2}$.
\end{proof}

Combining \cref{claim:oneplusepskmeans} with the coreset construction of \cref{cor:small-streaming-coreset}, we obtain:

\kmedianHamfull
\begin{proof}
By \cref{cor:small-streaming-coreset}, we have the following:
\[
\begin{aligned}
\begin{aligned}
f(\smalldelta) &= \Oh\!\left(\eps^{-6}k \log k \log (k/\eps) \right)\\
s(\smalldelta) &= \tOh_\eps(\ell k + k^2)\\
t(\smalldelta) &= \tOh_\eps(\ell + k^2)
\end{aligned}
\qquad\qquad
\begin{aligned}
f(\bigdelta) &= \Oh\!\left(\eps^{-6}k \log k \log (kn/\eps) \right)\\
s(\bigdelta) &= \tOh_\eps(\ell k + k^2 \log^2 n)\\
t(\bigdelta) &= \tOh_\eps(\ell + k^2 \log^2 n)
\end{aligned}
\end{aligned}
\]
We hence obtain by \cref{claim:oneplusepskmeans} that the space complexity of the algorithm is
\[
\Oh(s(\smalldelta) \cdot \log n + s(\bigdelta)) = \tOh_\eps((\ell k + k^2) \cdot \log n + \ell k + k^2 \log^2 n) = \tOh_\eps(\ell k + k^2)
\]
and the update time is 
\[
\begin{aligned}
&\Oh((t(\smalldelta) + 
\ell \cdot f(\smalldelta) \cdot k^{f(\smalldelta)}) \cdot \log n + t(\bigdelta) + k \ell \cdot f(\bigdelta)) = \\
&\tOh_\eps\Bigl(\ell + k^2
   + \ell \cdot (\eps^{-6}k \log k \log \tfrac{k}{\eps})
     \cdot k^{\Oh(\eps^{-6}k \log k \log \frac{k}{\eps})} \log n \\
&\qquad {}+ k \ell \cdot (\eps^{-6}k \log k \log \tfrac{kn}{\eps}) \Bigr) = \\
&\tOh_\eps(\ell \cdot 2^{\poly(\eps^{-1},k)}).
\qedhere
\end{aligned}
\]

\end{proof}

\section[Proof of Theorem~\ref{thm:small-coreset-offline}]{Proof of \cref{thm:small-coreset-offline}}
\label{sec:coreset-weighted}

\label{sec:hv}

In this section, we modify the algorithm of Huang and Vishnoi~\cite{HuangV20} to show \cref{thm:small-coreset-offline}.
%\smallcoresetoffline*
Our proof adapts their algorithm for \emph{unweighted} inputs to work on \emph{weighted} ones. The algorithm consists of two rounds referred to as importance sampling, but we only modify the first round, as the second round can already handle weighted inputs (see \cref{thm:reduction}). We give pseudocode of the first round in Algorithm~\ref{vh-alg}.

\begin{fact}[{\cite{HuangV20}}]
	\label{thm:reduction}
	Let $X$ be a set of $n$ points in $(\R^\ell, \eucl)$, $\eps,\delta \in (0,1/2)$ and $k$ be a positive integer. 
	Suppose that Algorithm~\ref{vh-alg} constructs an $\eps$-coreset for $k$-means clustering for $X$ of size $\Oh\left(\eps^{-25}  k^5 \log \frac{k}{\delta}\right)$.
	Then with probability at least $1-\delta/2$, the second round of the Huang--Vishnoi algorithm~\cite{HuangV20} outputs an $\Oh(\eps)$-coreset for $k$-means clustering for $X$ of size $\Oh\left(\eps^{-6}k \log k \log \frac{k}{\eps \delta} \right)$, and uses $\Oh(n\ell)$ time and space.
\end{fact}

\begin{algorithm}[h]
	\KwIn{An unweighted set $X$ of $n$ points in $\R^\ell$, $\eps, \delta \in (0,1/2)$ and an integer $k \geq 1$.}
	\KwOut{A point set $S\subseteq \R^\ell$ together with a weight function $\weight : S\rightarrow \R_{\geq 0}$.} 
	\tcc{The first importance sampling stage}
	$N_1 \leftarrow \Oh\left(\eps^{-25}  k^5 \log \frac{k}{\delta}  \right)$\;
	\label{line:constantfactor} Let $C^\star\subseteq \R^\ell$ (a set of $k$ centers) be a $\gamma$-approximation solution to the $k$-means clustering problem over $(X,\weight)$ ($\gamma = \Oh(1)$)\;
	For $x\in X$, let $c^\star(x)$ be the closest point to $x$ in $C^\star$ (ties are broken arbitrarily). For $c\in C^\star$, let $X_c$ be the set of points $x\in X$ with $c^\star(x)=c$\;
	For $x\in X$, let $\sigma_1(x)\leftarrow 2^{6} \gamma^2 \left(\frac{\euclsq(x,c^\star(x))}{ \nu(X,C^\star)}+\frac{1}{\absolute{X_{c^\star(x)}}}\right)$\;
	Pick a non-uniform random sample $D_1$ of $N_1$ points from $X$, where each $x\in X$ is selected with probability $\frac{\sigma_1(x)}{\sum_{y\in X}\sigma_1(y)}$.
	For $x\in D_1$, let $u(x)\leftarrow \frac{\sum_{y\in X}\sigma_1(y)}{\absolute{D_1}\cdot \sigma_1(x)}$\;
	Output $(D_1,u)$.
	\caption{The first round of Huang and Vishnoi's coreset construction \cite{HuangV20}}
    \label{vh-alg}
\end{algorithm}

We show that the importance sampling framework can be extended to handle weighted inputs by reducing a weighted input to an equivalent unweighted one. Specifically, if the total weight of the set is an integer
$W$, the corresponding unweighted input consists of $W$ points. Since $W$ may be unbounded, a direct application of this reduction could lead to an arbitrarily large running time. We demonstrate how to efficiently execute the algorithm directly on the weighted input, achieving small time and space complexity.

We start by introducing the following notations: given a weighted set $(X, \weight)$ in $\R^\ell$ of size~$n$, where all weights are integers, we define an unweighted set $T(X) \subset \R^{\ell + W}$ by expanding each point $x \in X$ into $\weight(x)$ distinct copies, obtained via small perturbations along additional coordinates.
Formally, define $T(X) = \{\, x_1^{(1)}, \dots, x_1^{(\weight(x_1))}, \dots, x_n^{(1)}, \dots, x_n^{(\weight(x_n))} \,\}$, 
where for each $i \in \{1,\dots, n\}$ and $j \in \{1, \dots, \weight(x_i)\}$, the point $x_i^{(j)} \in \R^{\ell+W}$ is defined so that:
\smallskip
\begin{enumerate}
	\item its first $\ell$ coordinates are identical to those of $x_i$;
	\item its $(\ell + \sum_{h=1}^{i-1} \weight(x_h) + j)$-th coordinate equals $\frac {\eps^{1/2} \eta}{\sqrt{6W}}$, where $\eta$ is the lower bound on the smallest $\eucl$ distance between two points in $X$, which is assumed to be known;
	\item all remaining coordinates are zero.
\end{enumerate} 
\smallskip
Furthermore, let $\alpha: \R^{\ell} \to \R^{\ell+W}$ be the function that appends zeros on the last $W$ dimensions, and let $\pi : \R^{\ell + W} \to \R^{\ell}$ be the projection on the first $\ell$ coordinates. We extend the definition of $\alpha$ and $\pi$ so that they respect multiplicities: if a set $S$ contains multiple points that share the same image under $\pi$, then $\pi(S)$ is interpreted as a weighted set, where the weight of each point is the total multiplicity of its preimages. Note that in particular, $\pi(T(X)) = (X,\weight)$. We have the following claim.
 
 \begin{lemma}\label{clm:projection-coreset}
 	If $(U, \weight_U)$ is an $\eps$-coreset for $k$-means clustering on $T(X)$ under the $\eucl$ distance and $U \subseteq T(X)$, then $\pi(U)$ is an $\Oh(\eps)$-coreset for $k$-means clustering on $X$ under the $\eucl$ distance.
 \end{lemma}

\begin{proof}
	First, recall that by \cref{clm:weight-of-coresets}, we have $\weight_U(U) \leq 6W$. 
    Consider arbitrary points $c_1, \dots, c_k \in \R^\ell$. Since $\absolute{X} \geq k+1$ by the condition of the theorem, there are $x,y \in X$ and $i$ such that $x$ and $y$ are both closest to $c_i$. Hence, 
\begin{align*}
&\eta^2 \leq \euclsq(x,y) \le (\eucl(x,c_i)+\eucl(c_i,y))^2 \\
&\leq 2 (\euclsq(x, c_i) + \euclsq(c_i, y)) \leq 2 \sum_{p \in X}\weight(p) \euclsq(p, \{c_1, \dots, c_k \}).
\end{align*}
    Therefore, we have
 \begin{align*}
\sum_{u \in U} w_U(u) \cdot & \euclsq(\pi(u), \{c_1, \dots, c_k \}) \\
&= \sum_{u \in U} w_U(u) \cdot \euclsq(u, \{\alpha(c_1), \dots, \alpha(c_k) \}) - \eta^2 \eps \frac{\weight_U(U)}{6W} \\
&\in (1 \pm \eps)\sum_{p' \in T(X)} \euclsq(p', \{\alpha(c_1), \dots, \alpha(c_k) \}) \pm \eta^2 \eps\\
&\in  (1 \pm \eps)\sum_{p' \in T(X)} \euclsq(\pi(p'), \{c_1, \dots, c_k \})  \pm (1 \pm \eps) \eta^2 \eps \pm \eta^2 \eps \\ 
&\in (1 \pm \Oh(\eps))\sum_{p \in X}\weight(p) \cdot \euclsq(p, \{c_1, \dots, c_k \}).\qedhere
 \end{align*}
This concludes the proof of the lemma. 
\end{proof}

However, computing a coreset of $T(X)$ by using Algorithm~\ref{vh-alg} would be too costly in both time and space, as this coreset is a subset of a $(\ell+W)$-dimensional space. Instead, in \cref{vh-alg-mod} we show an efficient implementation of Algorithm~\ref{vh-alg} for input $T(X)$.

\begin{algorithm}
	\caption{A modification of the first round of the Huang--Vishnoi algorithm~\cite{HuangV20}}
	\label{vh-alg-mod}
	\KwIn{A weighted set $X$ of $n$ points in $\R^\ell$, $\eps, \delta \in (0,1/2)$, and an integer $k \geq 1$.}
	\KwOut{A point set $S\subseteq \R^\ell$ together with a weight function $\weight:S\rightarrow \R_{\geq 0}$.}
	$N_1 \leftarrow \Oh\left(\eps^{-25}  k^5 \log \frac{k}{\delta}  \right)$\;
	Let $C^\star\subseteq \R^\ell$ (a set of $k$ centers) be a $\gamma$-approximation solution to the $k$-means clustering problem over $(X,\weight)$ ($\gamma = \Oh(1)$)\;\label{it:constantapprox}
	For $x\in X$, let $c^\star(x)$ be the closest point to $x$ in $C^\star$ (ties are broken arbitrarily). For $c\in C^\star$, let $X_c$ be the set of points $x\in X$ with $c^\star(x)=c$\;
	For $x\in X$, let $\sigma_1^T(x) = 64 \gamma^2 \left(\frac{\euclsq(x,c^\star(x)) + \frac {\eps \eta^2} {6W}}{ \nu(X,C^\star) +  \frac {\eps \eta^2} {6}}+\frac{1}{\weight(X_{c^\star(x)})}\right)$\;%, and for each $j \in \{1, \dots, w(x)\},$ set $\sigma_1^T(x^{(j)}) = \sigma_1(x)$\;
	Pick a non-uniform random sample $D_1'$ of $N_1$ points from $T(X)$, where each $x\in T(X)$ is selected with probability $\frac{\sigma_1^T(\pi(x))}{\sum_{y\in X} \weight(y)\sigma_1^T(\pi(y))}$. For $x\in D_1'$, let $u(x)\leftarrow \frac{\sum_{y\in X} \weight(y)\sigma_1(y)}{\absolute{D_1'}\cdot \sigma_1(x)}$\;
	%
	%For each $x\in X$, let $u(x)\leftarrow \absolute{\{j \in [w(x)]:x^{(j)} \in D_1' \}} \frac{\sum_{y\in X}\sigma_1(y)}{\absolute{D_1'}\cdot \sigma_1(x)}$, and let $D_1 = \{x \in X: u(x) \neq 0\}$\;
	%
	%For each $c\in C^\star$, let $u(c)\leftarrow (1+10\eps)\absolute{X_c}-\sum_{x\in D_1\cap X_c} u(x)$\;
	%
	
	Output $(\pi(D_1'),u)$.
\end{algorithm}
 
 \begin{lemma}\label{lem:similar-outputs}
 	Let $(X,\weight)$ be a weighted set of $\R^\ell$, let $U$ be the output of Algorithm~\ref{vh-alg} with input $T(X)$, and let $V$ be the output of Algorithm~\ref{vh-alg-mod} with input $(X,\weight)$ using the same randomness. We have $V = \pi(U)$.
 \end{lemma}
 \begin{proof}
 	To show the lemma, we prove that throughout each step of Algorithm~\ref{vh-alg-mod}, we simulate the behavior of Algorithm~\ref{vh-alg} using the two following observations.
 	
 	First, let $C^\star$ be a $\gamma$-approximation solution of the $k$-means clustering problem over~$(X,\weight)$. We have $\nu(T(X), \alpha(C^\star)) = \nu(X, C^\star) + \frac{\eta^2 \eps}{6}$. Furthermore, since $\absolute{X} \geq k+1$ by the condition of the theorem, there are $x,y \in X$ and $c \in C^\star$ such that $x$ and $y$ both get served by a single center $c$. Hence, 
\[\eta^2 \leq \euclsq(x,y) \le (\eucl(x,c)+\eucl(c,y))^2 \leq 2 (\euclsq(x, c) + \euclsq(c, y)) \leq 2 \nu(X, C^\star),\] 
which leads to $\nu(T(X), \alpha(C^\star))\leq 3 \nu(X, C^\star)$. Finally, since the cost of an optimal clustering of $T(X)$ is at least the cost of an optimal clustering of $X$, we know that $\alpha(C^\star)$ is a $3 \gamma$-approximation solution for $k$-means clustering on $T(X)$. As a result, it suffices to compute a constant-factor approximate solution for $k$-means clustering on $X$ to obtain a constant-factor approximation solution for $k$-means on~$T(X)$.
 	
 	Secondly, let $x \in X$, $i \leq \weight(x)$, and $C^\star \subseteq X$ be a subset of $X$ of size $k$. Define $c^\star(x)$ to be the closest point to $x$ in $C^\star$, $c_\alpha^\star(x^{(i)})$ the closest point to $x^{(i)}$ in $\alpha(C^\star)$, $X_{c}$ a set of points $x \in X$ such that $c^\star(x) = c$, and $T(X)_{c}$ a set of points $x \in T(X)$ such that $c_\alpha^\star(x) = c$. We have $c_\alpha^\star(x^{(i)}) =\alpha(c^\star(x))$, and thus $\weight(X_{c^\star(x)}) =  \absolute{T(X)_{c_\alpha^\star(x^{(i)})}}$; and
\[\frac{\euclsq(x^{(i)},c_\alpha^\star(x^{(i)}))}{ \nu(T(X),\alpha(C^\star))} = \frac{\euclsq(x,c^\star(x)) +  \frac {\eps \eta^2} {6W}}{\nu(X,C^\star) +  \frac {\eps \eta^2}{6}}\]
which shows that \cref{vh-alg-mod} indeed computes $\sigma_1$ for each point in $T(X)$ by computing $\sigma_1^T$ for each point in $X$. Finally, since we assume that we use the same randomness for the two algorithms, the sampling process samples the same points, which finishes the proof. \qedhere
\end{proof}

The correctness of \cref{thm:small-coreset-offline} follows from \cref{thm:reduction}, \cref{clm:projection-coreset} and \cref{lem:similar-outputs}. We now show how to implement the algorithm efficiently using small space. First, note that the points of $T(X)$ can be encoded efficiently in $\Oh(\ell \log n)$ space by storing the first $\ell$ coordinates of each point, and the index of the only non-zero coordinate. We use the algorithm of Mettu and Plaxton~\cite{MettuPlaxtonApprox} to implement Line~\ref{it:constantapprox} of \cref{vh-alg-mod}:
 
\begin{fact}[{\cite{MettuPlaxtonApprox}}]\label{fact:constantfactor}
Let $X$ be a weighted set of $n$ points in $(\R^\ell, \eucl)$. There exists an $\Oh(1)$-approximation algorithm that solves the $k$-means clustering problem on $X$ in $\Oh(n  \log W \cdot \max(k,\log n))$ time and space, where $W$ is the total weight of the points. 
\end{fact}

Finally, we show how to implement the sampling process efficiently.

\begin{lemma}\label{lm:sampling}
Let $\mathcal{D}$ be a distribution over $T(X)$ such that, for every $i \in \{1, \dots, n\}$, the samples $x_i^{(1)}, \dots, x_i^{(\weight(x_i))}$ share the same sampling rate $p_i$. Then, there exists an algorithm that can draw $N_1$ samples from $T(X)$ according to $\mathcal{D}$ in $\Oh(n + N_1 \log n)$ time and space.
\end{lemma}
\begin{proof}
We apply a so-called inverse transform procedure, see for example \cite[Chapter 2]{DBLP:books/sp/Devroye86}. Define $r_0 = 0$ and $r_j = \sum_{i=1}^j \weight(x_i) \cdot p_i$. Next, draw a random real number $r$ in $[0,1]$ and let $h$ be the integer such that $r \in [r_{h-1}, r_{h})$. Finally, set $f := \lfloor \frac{r-r_{h-1}}{p_h}\rfloor + 1$. The algorithm outputs a point $x_h^{(f)}$, and it is straightforward to verify that the point is indeed sampled with probability $p_h$.
\end{proof}
 
This concludes the proof of \cref{thm:small-coreset-offline}.

\begin{credits}
\subsubsection{\discintname}
The authors have no competing interests to declare that are relevant to the content of this article.
\end{credits}
\bibliographystyle{splncs04}
\bibliography{main}

@INPROCEEDINGS{10353181,
  author={Cohen-Addad, Vincent and Woodruff, David P. and Zhou, Samson},
  booktitle={FOCS}, 
  title={Streaming {E}uclidean k-median and k-means with o(log n) Space}, 
  year={2023},
  volume={},
  number={},
  pages={883-908},
  doi={10.1109/FOCS57990.2023.00057}}

@article{10.1093/bib/bby090,
    author = {Zou, Quan and Lin, Gang and Jiang, Xingpeng and Liu, Xiangrong and Zeng, Xiangxiang},
    title = {Sequence clustering in bioinformatics: an empirical study},
    journal = {Briefings in Bioinformatics},
    volume = {21},
    number = {1},
    pages = {1-10},
    year = {2020},
    month = {01},
    issn = {1477-4054},
    doi = {10.1093/bib/bby090}
}

@INPROCEEDINGS{9317938,
  author={Charalampopoulos, Panagiotis and Kociumaka, Tomasz and Wellnitz, Philip},
  booktitle={FOCS}, 
  title={Faster Approximate Pattern Matching: A Unified Approach}, 
  year={2020},
  volume={},
  number={},
  pages={978-989},
  doi={10.1109/FOCS46700.2020.00095}}

@article{10.1007/s10994-024-06540-z,
author = {Jiang, Shaofeng H. -C. and Krauthgamer, Robert and Lou, Jianing and Zhang, Yubo},
title = {Coresets for kernel clustering},
year = {2024},
issue_date = {Aug 2024},
address = {USA},
volume = {113},
number = {8},
issn = {0885-6125},
doi = {10.1007/s10994-024-06540-z},
journal = {Mach. Learn.},
month = {apr},
pages = {5891-5906},
numpages = {16}
}

@inproceedings{DBLP:conf/icalp/Bhattacharya023,
  author       = {Sudatta Bhattacharya and
                  Michal Kouck{\'{y}}},
  title        = {Streaming k-Edit Approximate Pattern Matching via String Decomposition},
  booktitle    = {{ICALP} 2023},
  pages        = {22:1--22:14},
  year         = {2023},
  doi          = {10.4230/LIPICS.ICALP.2023.22},
  bibsource    = {dblp computer science bibliography, https://dblp.org}
}

@article{10.1093/bib/bbs035,
    author = {Li, Weizhong and Fu, Limin and Niu, Beifang and Wu, Sitao and Wooley, John},
    title = {Ultrafast clustering algorithms for metagenomic sequence analysis},
    journal = {Briefings in Bioinformatics},
    volume = {13},
    number = {6},
    pages = {656-668},
    year = {2012},
    month = {11},
    issn = {1467-5463},
    doi = {10.1093/bib/bbs035}
}

@inproceedings{10.1145/3313276.3316350,
author = {Makarychev, Konstantin and Makarychev, Yury and Razenshteyn, Ilya},
title = {Performance of {J}ohnson-{L}indenstrauss transform for k-means and k-medians clustering},
year = {2019},
isbn = {9781450367059},
doi = {10.1145/3313276.3316350},
booktitle = {STOC},
pages = {1027-1038},
numpages = {12}
}

@article{10.1093/comjnl/bxm048,
    author = {Marx, D\'{a}niel},
    title = {Parameterized Complexity and Approximation Algorithms},
    journal = {The Computer Journal},
    volume = {51},
    number = {1},
    pages = {60-78},
    year = {2007},
    month = {07},
    abstract = {Approximation algorithms and parameterized complexity are usually considered to be two separate ways of dealing with hard algorithmic problems. In this paper, our aim is to investigate how these two fields can be combined to achieve better algorithms than what any of the two theories could offer. We discuss the different ways parameterized complexity can be extended to approximation algorithms, survey results of this type and propose directions for future research.},
    issn = {0010-4620},
    doi = {10.1093/comjnl/bxm048},
    -eprint = {https://academic.oup.com/comjnl/article-pdf/51/1/60/1357220/bxm048.pdf},
}

@InProceedings{10.1007/978-3-642-03784-9_23,
author="Amir, Amihood
and Landau, Gad M.
and Na, Joong Chae
and Park, Heejin
and Park, Kunsoo
and Sim, Jeong Seop",
title="Consensus Optimizing Both Distance Sum and Radius",
booktitle="SPIRE",
year="2009",
-publisher="Springer Berlin Heidelberg",
-address="Berlin, Heidelberg",
pages="234--242",
abstract="The consensus string problem is finding a representative string (consensus) of a given set {\$}{\backslash}mathbb{\{}S{\}}{\$}of strings. In this paper we deal with the consensus string problems optimizing both distance sum and radius, where the distance sum is the sum of (Hamming) distances from the strings in {\$}{\backslash}mathbb{\{}S{\}}{\$}to the consensus and the radius is the longest (Hamming) distance from the strings in {\$}{\backslash}mathbb{\{}S{\}}{\$}to the consensus. Although there have been results considering either distance sum or radius, there have been no results considering both as far as we know.",
isbn="978-3-642-03784-9"
}

@article{AMIR2013371,
title = {On the hardness of the Consensus String problem},
journal = {Information Processing Letters},
volume = {113},
number = {10},
pages = {371-374},
year = {2013},
issn = {0020-0190},
doi = {10.1016/j.ipl.2013.02.016},
author = {Amihood Amir and Haim Paryenty and Liam Roditty}
}

@InProceedings{10.1007/978-3-032-05228-5_12,
author="Gabory, Est{\'e}ban
and Bulteau, Laurent
and Fici, Gabriele
and Verbeek, Hilde",
title="String Consensus Problems with Swaps and Substitutions",
doi          = {10.1007/978-3-032-05228-5\_12},
booktitle="SPIRE",
year="2025",
pages="133--147",
isbn="978-3-032-05228-5"
}

@InProceedings{10.1007/978-3-642-13509-5_28,
author="Lee, Taehyung
and Na, Joong Chae
and Park, Heejin
and Park, Kunsoo
and Sim, Jeong Seop",
title="Finding Optimal Alignment and Consensus of Circular Strings",
doi          = {10.1007/978-3-642-13509-5\_28},
booktitle="CPM",
year="2010",
pages="310--322",
isbn="978-3-642-13509-5"
}

@article{DBLP:journals/algorithmica/AmirPR16,
  author       = {Amihood Amir and
                  Haim Paryenty and
                  Liam Roditty},
  title        = {Configurations and Minority in the String Consensus Problem},
  journal      = {Algorithmica},
  volume       = {74},
  number       = {4},
  pages        = {1267--1292},
  year         = {2016},
  doi          = {10.1007/S00453-015-9996-7},
  bibsource    = {dblp computer science bibliography, https://dblp.org}
}

@InProceedings{10.1007/978-3-319-07566-2_1,
author="Amir, Amihood
and Ficler, Jessica
and Roditty, Liam
and Shalom, Oren Sar",
title="On the Efficiency of the {H}amming {C}-Centerstring Problems",
 doi          = {10.1007/978-3-319-07566-2\_1},
booktitle="CPM",
year="2014",
pages="1--10",
isbn="978-3-319-07566-2"
}

@article{GASIENIEC2004289,
title = {Approximation algorithms for {H}amming clustering problems},
journal = {Journal of Discrete Algorithms},
volume = {2},
number = {2},
pages = {289-301},
year = {2004},
issn = {1570-8667},
doi = {10.1016/S1570-8667(03)00079-0},
author = {Leszek G\k{a}sieniec and Jesper Jansson and Andrzej Lingas}
}

@article{bulteau_et_al:LIPIcs.MFCS.2018.1,
  author       = {Laurent Bulteau and
                  Markus L. Schmid},
  title        = {Consensus Strings with Small Maximum Distance and Small Distance Sum},
  journal      = {Algorithmica},
  volume       = {82},
  number       = {5},
  pages        = {1378--1409},
  year         = {2020},
  doi          = {10.1007/S00453-019-00647-9},
  bibsource    = {dblp computer science bibliography, https://dblp.org}
}

@article{DBLP:journals/eatcs/BulteauHKN14,
  author       = {Laurent Bulteau and
                  Falk H{\"{u}}ffner and
                  Christian Komusiewicz and
                  Rolf Niedermeier},
  title        = {Multivariate Algorithmics for {NP}-Hard String Problems},
  journal      = {Bull. {EATCS}},
  volume       = {114},
  year         = {2014},
  -url          = {http://eatcs.org/beatcs/index.php/beatcs/article/view/310},
  timestamp    = {Fri, 12 Feb 2021 13:39:57 +0100},
  biburl       = {https://dblp.org/rec/journals/eatcs/BulteauHKN14.bib},
  bibsource    = {dblp computer science bibliography, https://dblp.org}
}

@article{DELAHIGUERA200039,
title = {Topology of strings: Median string is {NP}-complete},
journal = {Theoretical Computer Science},
volume = {230},
number = {1},
pages = {39-48},
year = {2000},
issn = {0304-3975},
doi = {https://doi.org/10.1016/S0304-3975(97)00240-5},
author = {Colin de la Higuera and Francisco Casacuberta}
}

@InProceedings{10.1007/3-540-44888-8_23,
author="Nicolas, Fran{\c{c}}ois
and Rivals, Eric",
title="Complexities of the Centre and Median String Problems",
booktitle="CPM",
year="2003",
-publisher="Springer Berlin Heidelberg",
-address="Berlin, Heidelberg",
pages="315--327",
abstract="Given a finite set of strings, the median string problem consists in finding a string that minimizes the sum of the distances to the strings in the set. Approximations of the median string are used in a very broad range of applications where one needs a representative string that summarizes common information to the strings of the set. It is the case in Classification, in Speech and Pattern Recognition, and in Computational Biology. In the latter, Median String is related to the key problem of Multiple Alignment. In the recent literature, one finds a theorem stating the NP-completeness of the median string for unbounded alphabets. However, in the above mentioned areas, the alphabet is often finite. Thus, it remains a crucial question whether the median string problem is NP-complete for finite and even binary alphabets. In this work, we provide an answer to this question and also give the complexity of the related centre string problem. Moreover, we study the parametrized complexity of both problems with respect to the number of input strings.",
isbn="978-3-540-44888-4"
}

@article{KOHONEN1985309,
title = {Median strings},
journal = {Pattern Recognition Letters},
volume = {3},
number = {5},
pages = {309-313},
year = {1985},
issn = {0167-8655},
doi = {10.1016/0167-8655(85)90061-3},
author = {Teuvo Kohonen}
}

@inproceedings{DBLP:conf/innovations/Chakraborty0K23,
  author       = {Diptarka Chakraborty and
                  Debarati Das and
                  Robert Krauthgamer},
  title        = {Clustering Permutations: New Techniques with Streaming Applications},
  booktitle    = {ITCS},
  series       = {LIPIcs},
  volume       = {251},
  pages        = {31:1--31:24},
  -publisher    = {Schloss Dagstuhl - Leibniz-Zentrum f{\"{u}}r Informatik},
  year         = {2023},
  -url          = {https://doi.org/10.4230/LIPIcs.ITCS.2023.31},
  doi          = {10.4230/LIPICS.ITCS.2023.31},
  timestamp    = {Wed, 21 Aug 2024 22:46:00 +0200},
  biburl       = {https://dblp.org/rec/conf/innovations/Chakraborty0K23.bib},
  bibsource    = {dblp computer science bibliography, https://dblp.org}
}

@inproceedings{doi:10.1137/1.9781611976465.48,
  author       = {Diptarka Chakraborty and
                  Debarati Das and
                  Robert Krauthgamer},
  -editor       = {D{\'{a}}niel Marx},
  title        = {Approximating the Median under the {U}lam Metric},
  booktitle    = {SODA},
  pages        = {761--775},
  -publisher    = {{SIAM}},
  year         = {2021},
  -url          = {https://doi.org/10.1137/1.9781611976465.48},
  doi          = {10.1137/1.9781611976465.48},
  timestamp    = {Thu, 15 Jul 2021 13:49:01 +0200},
  biburl       = {https://dblp.org/rec/conf/soda/ChakrabortyDK21.bib},
  bibsource    = {dblp computer science bibliography, https://dblp.org}
}

@inproceedings{10.1145/3717823.3718299,
author = {Cohen-Addad, Vincent and Grandoni, Fabrizio and Lee, Euiwoong and Schwiegelshohn, Chris and Svensson, Ola},
title = {A $(2+\varepsilon)$-Approximation Algorithm for Metric $k$-Median},
year = {2025},
isbn = {9798400715105},
doi = {10.1145/3717823.3718299},
booktitle = {STOC},
pages = {615-624},
numpages = {10}
}

@inproceedings{10.1145/3313276.3316307,
author = {Narayanan, Shyam and Nelson, Jelani},
title = {Optimal terminal dimensionality reduction in {E}uclidean space},
year = {2019},
isbn = {9781450367059},
doi = {10.1145/3313276.3316307},
booktitle = {STOC},
pages = {1064-1069},
numpages = {6}
}

@article{Achlioptas03,
  author       = {Dimitris Achlioptas},
  title        = {Database-friendly random projections: Johnson-{L}indenstrauss with binary
                  coins},
  doi          = {10.1016/S0022-0000(03)00025-4},
  journal      = {J. Comput. Syst. Sci.},
  volume       = {66},
  number       = {4},
  pages        = {671--687},
  year         = {2003}
}

@inproceedings{10.1145/3188745.3188828,
	author = {Mahabadi, Sepideh and Makarychev, Konstantin and Makarychev, Yury and Razenshteyn, Ilya},
	title = {Nonlinear dimension reduction via outer Bi-{L}ipschitz extensions},
	year = {2018},
	isbn = {9781450355599},
	doi = {10.1145/3188745.3188828},
	abstract = {We introduce and study the notion of *an outer bi-Lipschitz extension* of a map between Euclidean spaces. The notion is a natural analogue of the notion of *a Lipschitz extension* of a Lipschitz map. We show that for every map f there exists an outer bi-Lipschitz extension f′ whose distortion is greater than that of f by at most a constant factor. This result can be seen as a counterpart of the classic Kirszbraun theorem for outer bi-Lipschitz extensions. We also study outer bi-Lipschitz extensions of near-isometric maps and show upper and lower bounds for them. Then, we present applications of our results to prioritized and terminal dimension reduction problems, described next. We prove a *prioritized* variant of the Johnson–Lindenstrauss lemma: given a set of points X⊂ ℝd of size N and a permutation (”priority ranking”) of X, there exists an embedding f of X into ℝO(logN) with distortion O(loglogN) such that the point of rank j has only O(log3 + ε j) non-zero coordinates – more specifically, all but the first O(log3+ε j) coordinates are equal to 0; the distortion of f restricted to the first j points (according to the ranking) is at most O(loglogj). The result makes a progress towards answering an open question by Elkin, Filtser, and Neiman about prioritized dimension reductions. We prove that given a set X of N points in ℜd, there exists a *terminal* dimension reduction embedding of ℝd into ℝd′, where d′ = O(logN/ε4), which preserves distances ||x−y|| between points x∈ X and y ∈ ℝd, up to a multiplicative factor of 1 ± ε. This improves a recent result by Elkin, Filtser, and Neiman. The dimension reductions that we obtain are nonlinear, and this nonlinearity is necessary.},
	booktitle = {STOC},
	pages = {1088-1101},
	numpages = {14},
	keywords = {Prioritized Johnson-Lindenstrauss, Near-Isometric Maps, Metric Embedding, Dimension Reduction, Bi-Lipschitz Extension},
	location = {Los Angeles, CA, USA},
	-series = {STOC 2018}
}

@inproceedings{DBLP:conf/fsttcs/FominGS19,
	author       = {Fedor V. Fomin and
	Petr A. Golovach and
	Kirill Simonov},
	title        = {Parameterized k-Clustering: Tractability Island},
	booktitle    = {IARCS},
	series       = {LIPIcs},
	volume       = {150},
	pages        = {14:1--14:15},
	year         = {2019},
	doi          = {10.4230/LIPICS.FSTTCS.2019.14},
	bibsource    = {dblp computer science bibliography, https://dblp.org}
}

@article{DBLP:journals/corr/abs-2502-07653,
  author       = {Ragesh Jaiswal and
                  Amit Kumar and
                  Jatin Yadav},
  title        = {Robust-Sorting and Applications to {U}lam-Median},
  journal      = {CoRR},
  volume       = {abs/2502.07653},
  year         = {2025},
  doi          = {10.48550/ARXIV.2502.07653},
  eprinttype   = {arXiv},
  eprint       = {2502.07653},
  bibsource    = {dblp computer science bibliography, https://dblp.org}
}

@inproceedings{cohenaddad:hal-02360762,
	TITLE = {Inapproximability of Clustering in {Lp}-metrics},
	AUTHOR = {Cohen-Addad, Vincent and Srikanta, Karthik},
	BOOKTITLE = {FOCS},
    pages        = {519--539},
	YEAR = {2019},
	-MONTH = Nov,
	KEYWORDS = {graph embedding ; hardness of ap- proximation ; k-median ; k-means ; clustering},
	PDF = {https://hal.science/hal-02360762v2/file/main.pdf},
    doi          = {10.1109/FOCS.2019.00040},
	HAL_ID = {hal-02360762},
	HAL_VERSION = {v2},
}

@inproceedings{HuangV20,
	author       = {Lingxiao Huang and
	Nisheeth K. Vishnoi},
	title        = {Coresets for clustering in {E}uclidean spaces: importance sampling is
	nearly optimal},
	booktitle    = {STOC},
	pages        = {1416--1429},
	-publisher    = {{ACM}},
	year         = {2020},
	doi          = {10.1145/3357713.3384296},
	timestamp    = {Mon, 18 Dec 2023 07:33:37 +0100},
	biburl       = {https://dblp.org/rec/conf/stoc/HuangV20.bib},
	bibsource    = {dblp computer science bibliography, https://dblp.org}
}

@article{MettuPlaxtonApprox,
	author = {Mettu, Ramgopal R. and Plaxton, C. Greg},
	title = {Optimal Time Bounds for Approximate Clustering},
	year = {2004},
	issue_date = {June 2004},
	publisher = {Kluwer Academic Publishers},
	address = {USA},
	volume = {56},
	number = {1–3},
	issn = {0885-6125},
	doi = {10.1023/B:MACH.0000033114.18632.e0},
	journal = {Mach. Learn.},
	month = {jun},
	pages = {35-60},
	numpages = {26}
}

@article{DBLP:journals/siamcomp/Chen09,
	author       = {Ke Chen},
	title        = {On Coresets for k-Median and k-Means Clustering in Metric and {E}uclidean
	Spaces and Their Applications},
	journal      = {{SIAM} J. Comput.},
	volume       = {39},
	number       = {3},
	pages        = {923--947},
	year         = {2009},
	doi          = {10.1137/070699007},
	bibsource    = {dblp computer science bibliography, https://dblp.org}
}

@book{DBLP:books/sp/Devroye86,
  author       = {Luc Devroye},
  title        = {Non-Uniform Random Variate Generation},
  publisher    = {Springer},
  year         = {1986},
  doi          = {10.1007/978-1-4613-8643-8},
  isbn         = {978-1-4613-8645-2},
  bibsource    = {dblp computer science bibliography, https://dblp.org}
}

@inproceedings{10.1145/3166072.3166076,
	author = {Chappell, Timothy and Geva, Shlomo and Hogan, James},
	title = {K-Means Clustering of Biological Sequences},
	year = {2017},
	isbn = {9781450363914},
	doi = {10.1145/3166072.3166076},
	booktitle = {ADCS},
	articleno = {2},
    pages = {1--4}
}

@InProceedings{braverman_et_al:LIPIcs.APPROX-RANDOM.2019.62,
	author =	{Braverman, Vladimir and Feldman, Dan and Lang, Harry and Rus, Daniela},
	title =	{Streaming Coreset Constructions for $M$-Estimators},
	booktitle =	{APPROX/RANDOM},
	pages =	{62:1--62:15},
	ISBN =	{978-3-95977-125-2},
	ISSN =	{1868-8969},
	year =	{2019},
	volume =	{145},
	URN =		{urn:nbn:de:0030-drops-112778},
	doi =		{10.4230/LIPIcs.APPROX-RANDOM.2019.62}
}

@inproceedings{DBLP:conf/focs/OstrovskyR00,
	author       = {Rafail Ostrovsky and
	Yuval Rabani},
	title        = {Polynomial Time Approximation Schemes for Geometric k-Clustering},
	booktitle    = {FOCS},
	pages        = {349--358},
	-publisher    = {{IEEE} Computer Society},
	year         = {2000},
	doi          = {10.1109/SFCS.2000.892123},
	timestamp    = {Tue, 08 Jul 2025 16:42:26 +0200},
	biburl       = {https://dblp.org/rec/conf/focs/OstrovskyR00.bib},
	bibsource    = {dblp computer science bibliography, https://dblp.org}
}

@article{KARLOFF,
	title = {Fast algorithms for approximately counting mismatches},
	journal = {Information Processing Letters},
	volume = {48},
	number = {2},
	pages = {53-60},
	year = {1993},
	issn = {0020-0190},
	doi = {10.1016/0020-0190(93)90177-B},
	author = {Howard Karloff}
}

@InProceedings{kopelowitz_et_al:OASIcs.SOSA.2018.10,
	author =	{Kopelowitz, Tsvi and Porat, Ely},
	title =	{A Simple Algorithm for Approximating the Text-To-Pattern Hamming Distance},
	booktitle =	{SOSA},
	pages =	{10:1--10:5},
	series =	{OASIcs},
	ISBN =	{978-3-95977-064-4},
	ISSN =	{2190-6807},
	year =	{2018},
	volume =	{61},
	URN =		{urn:nbn:de:0030-drops-83089},
	doi =		{10.4230/OASIcs.SOSA.2018.10}
}

@inproceedings{DBLP:conf/esa/FischerGH025,
  author       = {Nick Fischer and
                  Elazar Goldenberg and
                  Mursalin Habib and
                  {Karthik {C. S.}}},
  -editor       = {Anne Benoit and
                  Haim Kaplan and
                  Sebastian Wild and
                  Grzegorz Herman},
  title        = {Hardness of Median and Center in the {U}lam Metric},
  booktitle    = {ESA},
  series       = {LIPIcs},
  volume       = {351},
  pages        = {111:1--111:17},
  -publisher    = {Schloss Dagstuhl - Leibniz-Zentrum f{\"{u}}r Informatik},
  year         = {2025},
  -url          = {https://doi.org/10.4230/LIPIcs.ESA.2025.111},
  doi          = {10.4230/LIPICS.ESA.2025.111},
  timestamp    = {Wed, 07 Jan 2026 17:31:38 +0100},
  biburl       = {https://dblp.org/rec/conf/esa/FischerGH025.bib},
  bibsource    = {dblp computer science bibliography, https://dblp.org}
}

@INPROCEEDINGS{10353074,
  author={Abbasi, Fateme and Banerjee, Sandip and Byrka, Jarosław and Chalermsook, Parinya and Gadekar, Ameet and Khodamoradi, Kamyar and Marx, Dániel and Sharma, Roohani and Spoerhase, Joachim},
  booktitle={FOCS}, 
  title={Parameterized Approximation Schemes for Clustering with General Norm Objectives}, 
  year={2023},
  volume={},
  number={},
  pages={1377-1399},
  doi={10.1109/FOCS57990.2023.00085}}

\appendix
\section{Omitted proofs}
\streaming*

\coresetanalysis*
\begin{proof} 
With high probability, the matrix defined in \cref{fact:JL} satisfies \cref{eq:jl} and the algorithm of \cref{fact:means_clustering} is correct. Furthermore, with probability $1-\delta$, the algorithm of \cref{thm:small-coreset-offline} is correct. Below, we condition on these assumptions. 
We start with the following claim: 

\begin{claim}\label{clm:prop-of-mu}
	For all $x \in X, y \in \Sigma^\ell$, we have
    \[(1-\eps)^2 \cdot \hd(x, y) \leq \euclsq(\mu(x), \mu(y)) \leq (1+\eps)^3 \cdot \hd(x, y).\]
\end{claim}
\begin{proof}

	By \cref{fact:embedding}, for all $x,y \in \Sigma^\ell$, we have 
	$\hd(x,y) \le \frac 1 {\frac r 2 - \frac \eps 6 r} \cdot \hd(\kappa(x),\kappa(y)) \le (1+\eps) \cdot \hd(x,y)$.
	By the definition of the Hamming distance, $\hd(\kappa(x),\kappa(y)) = \euclsq(\kappa(x), \kappa(y))$. Hence, by \cref{fact:terminal}, for $x \in X$, we have 
	\begin{align*}
		(1-\eps)^2 \cdot \hd(x,y) & \le \frac 1{\frac r 2 - \frac \eps 6 r} \cdot (1-\eps)^2 \euclsq(\kappa(x) , \kappa(y)) \\
		&\le \euclsq(\mu(x), \mu(y)) \\
        & \le 
        \frac 1{\frac r 2 - \frac \eps 6 r} \cdot  (1+\eps)^2 \cdot \euclsq(\kappa(x), \kappa(y)) \le (1+\eps) (1+\eps)^2 \cdot \hd(x,y).\qedhere
	\end{align*}
\end{proof}

Now, consider a set $C$ of $k$ points in $\Sigma^\ell$, a set $P \subseteq X$, and recall that $\nu(P,C) = \sum_{p \in P} \hd(p,C) \cdot \weight(p)$.
By \cref{clm:prop-of-mu}, for every $p \in P$ and $c \in C$ we have 
\[\hd(p, C) = \min_{c \in C} \hd(p,c) \le \frac{1}{(1-\eps)^2} \min_{c \in C} \euclsq(\mu(p), \mu(c))\]
On the other hand, we have $(1+\eps)^3 \cdot \hd(p,c) \ge \euclsq(\mu(p), \mu(c))$, and therefore 
\[\hd(p,C) = \min_{c \in C} \hd(p,c) \ge \frac{1}{ (1+\eps)^3} \min_{c \in C} \euclsq(\mu(p), \mu(c))\]
It follows that 

\begin{equation}
\begin{aligned}
&\frac{1}{(1+\eps)^3}
\sum_{p \in P} \weight(p)
\cdot \min_{c \in C}
\euclsq(\mu(p), \mu(c)) \le \\
&\nu(P,C) \le \\
&\frac{1}{(1-\eps)^2}
\sum_{p \in P} \weight(p)
\cdot \min_{c \in C}
\euclsq(\mu(p), \mu(c))
\end{aligned}
\label{eq:nu}
\end{equation}
This holds for any set $P \subseteq X$, in particular, for $P = X$:
\begin{align*}
\frac{1}{(1+\eps)^3} \sum_{x \in X} \min_{c \in C} \euclsq(\mu(x), \mu(c)) \le \nu(X,C) \le  \frac{1}{(1-\eps)^2} \sum_{x \in X}   \min_{c \in C} \euclsq(\mu(x), \mu(c))
\end{align*}
Now recall that $S$ is a preimage of $Q$. For $Q$, notice that by combining two coreset constructions sequentially, we obtain that $Q$ is a $((1+\eps)^2 -1)$-coreset of $\mu(X)$. Since $((1+\eps)^2 -1) = 2\eps + \eps^2 < 3\eps$, $Q$ is also a $3\eps$-coreset of $\mu(X)$. By \cref{def:coreset}, we have 
\begin{equation}
\begin{aligned}
	&\Absolute{\sum_{q \in Q} \weight_Q(q) \min_{c \in C} \euclsq(q, \mu(c)) - \sum_{x \in X} \min_{c \in C} \euclsq(x, \mu(c))} \\
	&\text{\qquad \qquad \qquad \qquad \qquad \qquad \qquad}\le 3\eps \sum_{x \in X} \min_{c \in C} \euclsq(\mu(x), \mu(c))
\end{aligned}
\label{eq:Q}
\end{equation}
Consequently,
\begin{align*}
\nu(S,C) & \le \frac{1}{(1-\eps)^2} \sum_{s \in S} \weight_S(s) \cdot \min_{c \in C} \euclsq(\mu(s),\mu(c)) & & \\
& = \frac{1}{(1-\eps)^2} \sum_{q \in Q} \weight_Q(q) \cdot \min_{c \in C} \euclsq(q,\mu(c)) & & \text{($S$ is a preimage of $Q$ under $\mu$)} \\
&  \le \frac{1+3\eps}{(1-\eps)^2}  \sum_{x \in X} \min_{c \in C} \euclsq(\mu(x),\mu(c)) & & \text{(\cref{eq:Q})} \\
& \le \frac{(1 + 3\eps)(1+\eps)^3}{(1-\eps)^2}  \cdot \nu(X,C) & & \text{(\cref{eq:nu})} \\
& \le (1 + \Oh(\eps))  \cdot \nu(X,C) & & \text{holds for $\eps \in (0, 1/2)$}.
\end{align*}
On the other hand,
\begin{align*}
\nu(S,C) &\ge \frac{1}{(1+\eps)^3} \sum_{s \in S} \weight_S(s) \cdot \min_{c \in C} \euclsq(\mu(s),\mu(c)) & &  \\
& = \frac{1}{(1+\eps)^3} \sum_{q \in Q} \weight_Q(q) \cdot \min_{c \in C} \euclsq(q,\mu(c)) & & \text{($S$ is a preimage of $Q$ under $\mu$)}\\
& \ge \frac{1-3\eps}{(1+\eps)^3} \sum_{x \in X} \min_{c \in C} \euclsq(\mu(x), \mu(c))  & & \text{(\cref{eq:Q})} \\
& \ge  \frac{(1-3\eps)(1-\eps)^2}{(1+\eps)^3} \cdot \nu(X,C) & & (\text{\cref{eq:nu}})\\
& \ge  (1 - 8\eps) \cdot \nu(X,C) & & \text{holds for $\eps \in (0, 1/2)$}.
\end{align*}
Hence, $(S, w_S)$ is indeed an $\Oh(\eps)$-coreset.
\end{proof}

\begin{lemma}\label{lem:integer-reweighting}
Let $(\X,d)$ be a metric space, let $(X,w)$ be a weighted set of $n$
points in $\X$, let $z,k$ be positive integers, and let
$\eps \in (0,1/2)$. Set $w_{\min}:=\min_{x\in X}w(x)$ and define an integer
weight function $w'\colon X\to\mathbb{Z}_{>0}$ by $w'(x):=\left\lfloor \frac{w(x)}{\eps\, w_{\min}}\right\rfloor$.
If $(S,\omega)$ is an $\eps$-coreset for
$(k,z)$-clustering of the integer-weighted set $(X,w')$, the rescaled set
$(S,\,\eps\, w_{\min}\,\omega)$ is a $2\eps$-coreset for
$(k,z)$-clustering of $(X,w)$.
\end{lemma}
\begin{proof}
By definition of the floor, $\tfrac{w(x)}{\eps w_{\min}}-1< w'(x)\le
\tfrac{w(x)}{\eps w_{\min}}$, and hence
\[
  (1-\eps)\,w(x)\ \le\ \eps\, w_{\min}\,w'(x)\ \le\ w(x)
  \qquad\text{for all }x\in X.
\]
Fix any $C\subseteq\mathcal{X}$ with $|C|\le k$. Multiplying by
$(d(x,C))^z\ge 0$ and summing over $x\in X$ gives
\begin{equation}\label{eq:reweight-sandwich}
  (1-\eps)\,\nu_{z,w}(X,C)\ \le\ \eps\, w_{\min}\,\nu_{z,w'}(X,C)
  \ \le\ \nu_{z,w}(X,C)
\end{equation}

Now let $(S,\omega)$ be an $\eps$-coreset for $(k,z)$-clustering of
$(X,w')$, and write $\omega^{\ast}:=\eps w_{\min}\,\omega$, so that
$\nu_{z,\omega^{\ast}}(S,C)=\eps w_{\min}\,\nu_{z,\omega}(S,C)$. We have:
\begin{align*}
(1-\eps)\,\nu_{z,w'}(X,C)\ \le\ \nu_{z,\omega}(S,C)\ \le\
  (1+\eps) \nu_{z,w'}(X,C) \text{\quad (Def.~\ref{def:coreset})} \\
(1-\eps)\eps w_{\min} \nu_{z,w'}(X,C)\ \le \eps w_{\min} \nu_{z,\omega}(S,C)\ \le\ (1+\eps) \eps \, w_{\min} \nu_{z,w'}(X,C)  \\
(1-\eps)^2\nu_{z,w}(X,C)\ \le \nu_{z,\omega^\ast}(S,C)\ \le\
  (1+\eps)  \nu_{z,w}(X,C) \text{\quad (Eq.~\ref{eq:reweight-sandwich})}  \\
\end{align*}
Finally, $(1-\eps)^2\ge 1-2\eps$ and $1+\eps\le 1+2\eps$,
hence $\lvert\nu_{z,\omega^{\ast}}(S,C)-\nu_{z,w}(X,C)\rvert\le
2\eps\,\nu_{z,w}(X,C)$ for every such $C$. Thus
$(S,\omega^{\ast})$ is a $2\eps$-coreset for $(k,z)$-clustering of
$(X,w)$ by \cref{def:coreset}.
\end{proof}

\end{document}